\documentclass[fleqn]{article}
\usepackage{amsmath}
\usepackage{amsthm}
\usepackage{amssymb}
\usepackage{enumerate}

\theoremstyle{remark}
\newtheorem*{cor}{Corollary}

\theoremstyle{definition}
\newtheorem*{defn}{Definition}

\theoremstyle{plain}
\newtheorem*{lem}{Lemma}

\usepackage[pdftex]{hyperref}

\begin{document}
\begin{center}
\Large
{\bf Local unitary invariants for multipartite quantum systems}
\end{center}
\vspace*{-.1cm}
\begin{center}
P\'eter Vrana
\end{center}
\vspace*{-.4cm} \normalsize
\begin{center}
Department of Theoretical Physics, Institute of Physics, Budapest University of\\ Technology and Economics, H-1521 Budapest, Hungary

\vspace*{.2cm}
(\today)
\end{center}

\begin{abstract}
A method is presented to obtain local unitary invariants for multipartite
quantum systems consisting of fermions or distinguishable particles. The
invariants are organized into infinite families, in particular, the generalization
to higher dimensional single particle Hilbert spaces is straightforward.
Many well-known invariants and their generalizations are also included.
\end{abstract}

\section{Introduction}

The possibility of entanglement between subsystems is a purely quantum mechanical
phenomenon, related to the nonlocal nature of the fundamental laws governing our
world. In order to understand quantum entanglement, one constructs functions on
the space of quantum states invariant under groups modelling local operations.
The two main approaches consider local unitary (LU) operations which describe
local transformations that can be applied with probability one, or the SLOCC group
(stochastic local operations and classical communication) corresponding to
transformations that can be done with nonzero probability\cite{Dur}.

Various results exist for both groups acting on composite quantum systems with
distinguishable constituents\cite{Virmani,Horodecki,Bengtsson}, but much less is known about the entanglement
of indistinguishable particles. Recent work revealed that an understanding of
fermionic entanglement can also provide us with information about the entanglement
of distinguishable subsystems\cite{Eckert,VL}. Motivated by this, in this paper we would like
to present a way to construct fermionic entanglement measures for pure states.

Generally, one looks for invariant functions that are polynomial in the
coefficients and their conjugates of the pure state with respect to a fixed
orthonormal basis. More abstractly (and independently of basis choices), one would
like to find subrepresentations of $S(V\oplus V^{*})$ isomorphic to the trivial
one or one dimensional subrepresentations of $S(V)$, where $V$ is a representation
of some group $G$ (either LU or SLOCC) and $S(\cdot)$ denotes the symmetric algebra
on a vector space.

Here, we take a slightly different approach, and use the projections to every
subrepresentation in $S(V)$ and associate invariants to them. Here $V$ is
the state space of a $k$-fermion system, and the group considered is the LU group,
i.e. the unitary group acting on the one-particle Hilbert space. This approach
has the advantage that the resulting formulae are independent of the dimension of
the one-particle state space.

The outline of the paper is as follows. In section~\ref{sec:sym} we recall how each graded
part of the symmetric algebra of a Hilbert space comes equipped with an invariant
inner product. In section~\ref{sec:invs} this inner product is utilized in order to associate
LU-invariants to every isotypic subspace of the space of degree $m$ polynomials in
the coefficients of a multi-fermion state. In section~\ref{sec:high} a special case is
considered, namely, the invariant associated to the subrepresentation containing
the weight space with highest weight. The method to obtain explicit formulae
is also presented here. In section~\ref{sec:ex} some examples are worked out illustrating
various features of our approach. In section~\ref{sec:SLOCC} the relationship between LU-invariants
generated this way and SLOCC-invariants is highlighted. In section~\ref{sec:dist} it is
briefly mentioned, how the fermionic invariants obtained can be used to construct
local unitary invariants for quantum systems with distinguishable constituents.
For the readers' convenience, a summary of some concepts from the representation
theory of the unitary groups can be found in the Appendix.

\section{The symmetric algebra of a Hilbert space}\label{sec:sym}

Throughout this section $\mathcal{H}$ denotes a finite dimensional complex Hilbert
space with inner product $\langle\cdot,\cdot\rangle$. We regard $\mathcal{H}$ as a
representation of $U(\mathcal{H})=\{\varphi:\mathcal{H}\to\mathcal{H}|\forall v\in\mathcal{H}:\|\varphi v\|=\|v\|\}$.

Let $S(\mathcal{H})$ denote the symmetric algebra of $\mathcal{H}$ that is, the
algebra of polynomials in vectors of $\mathcal{H}$. $S(\mathcal{H})$ has the
structure of a graded algebra, its degree $m$ homogenous subspace will be denoted
by $S^{m}(\mathcal{H})$. As $S^{1}(\mathcal{H})=\mathcal{H}$, and this subspace
generates $S(\mathcal{H})$ as a unital commutative algebra, we have that $U(\mathcal{H})$
acts on $S(\mathcal{H})$ with algebra automorphisms.

The inner product on $\mathcal{H}$ induces one on $S^{m}(\mathcal{H})$ by the
following requirement: for $u,v\in\mathcal{H}$ let $\langle u^{m},v^{m}\rangle=\langle u,v\rangle^{m}$.
This turns out to be equivalent to saying that for a unit vector $u\in\mathcal{H}$, $\|u^{m}\|=1$.
Clearly, this inner product will be preserved by the action of $U(\mathcal{H})$
on $S(\mathcal{H})$, restricted to each homogenous subspace. It is known from the
representation theory of the unitary groups that in this way each
$S^{m}(\mathcal{H})$ becomes an irreducible unitary representation of $U(\mathcal{H})$,
and hence the induced inner product is essentially the only one invariant under
this group action.

To be more explicit, if we fix an orthonormal basis $\{e_1,\ldots,e_d\}$ in
$\mathcal{H}$, then $S^{m}(\mathcal{H})$ is the space of degree $m$ homogenous
polynomials in the basis elements, and the degree $m$ monomials with coefficient
$1$ form a basis. These monomials are mutually orthogonal, but they are not
unit vectors. If $v=\sum_{i=1}^{d}\alpha_i e_i$ then
\begin{equation}
\begin{split}
v^{m} & = \sum_{i_1=1}^{d}\cdots\sum_{i_m=1}^{d}\alpha_{i_1}\alpha_{i_2}\cdots\alpha_{i_m}e_{i_1}\cdots e_{i_m}  \\
      & = \sum_{\substack{k_1,\ldots,k_d\ge 0  \\  k_1+\ldots+k_d=m}}\binom{m}{k_1,k_2,\ldots,k_d}\alpha_{1}^{k_1}\alpha_{2}^{k_2}\cdots\alpha_{d}^{k_d}e_{1}^{k_1}e_{2}^{k_2}\cdots e_{d}^{k_d}
\end{split}
\end{equation}
(where $\binom{m}{k_1,k_2,\ldots,k_d}$ is the multinomial coefficient) hence
\begin{equation}
\|v^{m}\|^{2}=\sum_{\substack{k_1,\ldots,k_d\ge 0  \\  k_1+\ldots+k_d=m}}{\binom{m}{k_1,k_2,\ldots,k_d}}^{2}|\alpha_{1}^{k_1}\alpha_{2}^{k_2}\cdots\alpha_{d}^{k_d}|^{2}\|e_{1}^{k_1}e_{2}^{k_2}\cdots e_{d}^{k_d}\|^{2}
\end{equation}
Comparing this with
\begin{equation}
(\|v\|^{2})^{m}=\left(\sum_{i=1}^{d}|\alpha_i|^{2}\right)^{m}=\sum_{\substack{k_1,\ldots,k_d\ge 0  \\  k_1+\ldots+k_d=m}}\binom{m}{k_1,k_2,\ldots,k_d}|\alpha_{1}^{k_1}\alpha_{2}^{k_2}\cdots\alpha_{d}^{k_d}|^{2}
\end{equation}
we conclude that
\begin{equation}\label{eq:norm}
\|e_{1}^{k_1}e_{2}^{k_2}\cdots e_{d}^{k_d}\|={\binom{m}{k_1,k_2,\ldots,k_d}}^{-1/2}
\end{equation}

\section{Invariants for multi-fermion systems}\label{sec:invs}

In this section, $\mathcal{H}$ will be a finite dimensional complex Hilbert
space, playing the role of the single-particle state space of a fermionic
quantum system of $k$ particles. If $n=\dim\mathcal{H}$, then the $k$-particle
Hilbert space is isomorphic to
\begin{equation}
\bigwedge^{k}\mathcal{H}\simeq\bigwedge^{k}\mathbb{C}^{n}
\end{equation}
and hence its dimension is $\binom{n}{k}$. This space also comes equipped with
an inner product induced from that of $\mathcal{H}$, and a unitary action of
$U(\mathcal{H})$ which models local unitary transformations of the $k$-particle states.

Now let us look at the symmetric algebra of the $k$-fermion state space. On its
homogenous subspaces $S^{m}(\bigwedge^{k}\mathcal{H})$ we have an action of
$U(\mathcal{H})$ which factors through $U(\bigwedge^{k}\mathcal{H})$ and an
inner product which is invariant under $U(\bigwedge^{k}\mathcal{H})$
hence also invariant under $U(\mathcal{H})$. This time the representation of
$U(\mathcal{H})$ is not irreducible, and $S^{m}(\bigwedge^{k}\mathcal{H})$
can be split into the orthogonal sum of $U(\mathcal{H})$-invariant subspaces
in a non-trivial way:
\begin{equation}
S^{m}(\bigwedge^{k}\mathcal{H})=\bigoplus_{\lambda} V_{\lambda}
\end{equation}
where $\lambda$ ranges over the partitions of $km$, and $V_{\lambda}$ is the
corresponding isotypic component of the representation. Interestingly, this
decomposition is independent of $n$ (apart from the vanishing of the
subrepresentations associated to partitions involving more than $n$ parts, but
for $n\ge km$ this certainly cannot happen). This is essentially due to the fact
that a degree $j$ symmetric polynomial in $n$ variables can be reconstructed
even if we only now its restriction to a subspace in which only $j$ variables
take nonzero values.

This decomposition allows us to introduce unitary invariants, one for each
isotypic subspace. Let $\psi\in\bigwedge^{k}\mathcal{H}$ be a $k$-fermion state
vector, and $\psi^{m}$ its $m$-th power which is an element of $S^{m}(\bigwedge^{k}\mathcal{H})$.
Let $P_{\lambda}:S^{m}(\bigwedge^{k}\mathcal{H})\to V_{\lambda}$ denote the orthogonal
projection. This commutes with the representation of $U(\mathcal{H})$, therefore
the value of $I_{\lambda}(\psi):=\langle\psi^m,P_{\lambda}\psi^m\rangle=\|P_{\lambda}\psi^m\|^{2}$ is
invariant:
\begin{equation}
\forall g\in U(\mathcal{H}):\langle (g\cdot\psi)^m,P_{\lambda}(g\cdot\psi)^m\rangle=\langle g\cdot(\psi^m),g\cdot(P_{\lambda}\psi^m)\rangle=\langle\psi^m,P_{\lambda}\psi^m\rangle
\end{equation}
Note that the number of linearly independent invariants is one less than the
number of nonvanishing isotypic components, because
\begin{equation}\label{eq:constraint}
\sum_{\lambda}\langle\psi^m,P_{\lambda}\psi^m\rangle=\langle\psi^m,\left(\sum_{\lambda}P_{\lambda}\right)\psi^m\rangle=\langle\psi^m,\psi^m\rangle=1
\end{equation}

Unfortunately, it is in general not an easy task to calculate the projections for
all these invariant subspaces for every value of $k$ and $m$, but some of them
are easy enough to be done by hand.

\section{Invariant subspaces with maximal highest weight}\label{sec:high}

Let us now fix an ordered orthonormal basis $(e_{1},\ldots,e_{n})$ in $\mathcal{H}$.
This also gives the isomorphisms $\mathcal{H}\simeq\mathbb{C}^{n}$, and
$U(\mathcal{H})\simeq U(n,\mathbb{C})$. The maximal torus $T$ which acts diagonally in
this basis is then identified with the subgroup of diagonal unitary matrices.
The set of one dimensional representations $T\to\mathbb{C}^{\times}$ is a commutative group
isomorphic to $\mathbb{Z}^{n}$. We will use the following identification:
\begin{equation}
\begin{split}
(r_1,r_2,\ldots,r_n):T\to\mathbb{C}^{\times}  \\
(r_1,r_2,\ldots,r_n)(diag(\lambda_1,\ldots,\lambda_n))=\prod_{i=1}^{n}\lambda_{i}^{r_i}
\end{split}
\end{equation}
On the set of $n$-tuples of integers we have the usual partial ordering:
$(r_1,r_2,\ldots,r_n)$ is called positive iff $r_1+\ldots+r_n=0$ and $r_1$, $r_1+r_2$,
\ldots, $r_1+r_2+\ldots+r_{n-1}$ are nonnegative, and $\lambda\ge\mu$ iff $\lambda-\mu$
is positive. A finite dimensional representation of $U(\mathcal{H})$, when restricted
to $T$, splits into one dimensional subrepresentations. The representations with
nonzero multiplicity are called weights, and a vector whose orbit under $T$ spans a
one dimensional subspace is called a weight vector. The isomorphism class of an irreducible
representation of $U(\mathcal{H})$ is determined by its highest weight.

For $I=\{i_1,i_2,\ldots,i_k\}$ where $1\le i_1<i_2<\ldots<i_k\le n$ let us introduce the
following notation:
\begin{equation}\label{eq:basis}
e_{I}=e_{i_1}\wedge e_{i_2}\wedge\ldots\wedge e_{i_k}=\frac{1}{\sqrt{k!}}\sum_{\pi\in S_{k}}\sigma(\pi)e_{i_\pi(1)}\otimes e_{i_\pi(2)}\otimes\ldots\otimes e_{i_\pi(k)}
\end{equation}
where $S_{k}$ is the symmetric group on $k$ elements, and $\sigma:s_{k}\to\{1,-1\}$
denotes the alternating representation. The set $\{e_{\{i_1,\ldots,i_k\}}|1\le i_1<i_2<\ldots<i_k\le n\}$
forms an orthonormal basis of $\bigwedge^{k}\mathcal{H}$, and therefore every $k$-fermion
pure state can be expressed uniquely as a linear combination of these vectors:
\begin{equation}
\psi=\sum_{I\in\binom{[n]}{k}}\psi_{I}e_{I}\textrm{ where }\sum_{I\in\binom{[n]}{k}}|\psi_{I}|^{2}=1
\end{equation}
(Here we have used the short notation $[n]=\{1,2,\ldots,n\}$ and $\binom{[n]}{k}$ denotes the set
of $k$-element subsets of $[n]$.) For each $m\in\mathbb{N}$, the $m$th power of $\psi$ is a vector
in $S^{m}(\bigwedge^{k}\mathcal{H})$:
\begin{equation}
\psi^{m}=\sum_{I_1,\ldots,I_m}\psi_{I_1}\psi_{I_2}\ldots\psi_{I_m}e_{I_1}e_{I_2}\ldots e_{I_m}
\end{equation}

We would like to find a vector in $S^{m}(\bigwedge^{k}\mathcal{H})$ which generates
an irreducible $U(\mathcal{H})$-representation. In general we cannot say much about
all the irreducible subrepresentations, but we always have one weight vector,
$e_{1,2,\ldots,k}^m$, corresponding to the highest weight, which is easily calculated
to be $(m,m,\ldots,m,0,\ldots,0)$ with $k$ nonzero entries. We now have that
$\langle U(\mathcal{H})e_{1,2,\ldots,k}^m\rangle:=W$ is irreducible. The next step
will be to find an orthonormal basis for $W$.

Our first goal will be to find a generating set for $W$ as a linear space, then we
can orthogonalize it to yield an orthonormal basis. To this end, we will use the
fact that $W$ is also an irreducible representation of $GL(n,\mathbb{C})$ whose action
on $S^{m}(\bigwedge^{k}\mathcal{H})$ is defined in the same way as that of $U(\mathcal{H})$.

In order to find a generating set which is easy to handle, we will look for one that
is the union of orbits under $S_n\le GL(n,\mathbb{C})$ (possibly up to a nonzero multiple)
which permutes the basis elements of $\mathcal{H}$. It turns out that we can
require also that the generating set consists of weight vectors. We will call sets
with these properties \emph{good}:
\begin{defn}
Let $S\subset S^{m}(\bigwedge^{k}\mathcal{H})$ be a subset, $e_1,\ldots,e_n$ an
orthonormal basis in $\mathcal{H}$ and $S_n\le U(\mathcal{H})$ the subgroup
which permutes these basis elements. The subset $S$ will be called \emph{good}
(with respect to this basis) if it has the following two properties:
\begin{enumerate}[1.]
\item The subset
\begin{equation}
\mathbb{C}S:=\bigcup_{w\in S}\mathbb{C}w\subseteq S^{m}(\bigwedge^{k}\mathcal{H})
\end{equation}
is fixed under the action of $S_n$.
\item If $v$ is an element of $S$ then if we write $v$ as a polynomial in the vectors
$\{e_{I}\}_{I\in\binom{[n]}{k}}$ then every index $i\in[n]$ appears the same number of
times in every term. Or equivalently: $v$ is a weight vector for the maximal torus fixing
the given orthonormal basis.
\end{enumerate}
\end{defn}
We can immediately see that $\{e_{I}^m\}_{I\in\binom{[n]}{k}}$ is the smallest
\emph{good} subset containing $e_{1,2,\ldots,k}^m$.

To reach every element in $W$, we will use the fact that $GL(n,\mathbb{C})$
is generated by matrices of the form $u_{ij}(s)=id+sE_{ij}$ where $E_{ij}$ is a matrix
with a $1$ at the intersection of the $i$th row and the $j$th column, and zeros everywhere
else. We need to know how these matrices act on the basis elements of $\bigwedge^{k}(\mathcal{H})$.
One can calculate using equation~(\ref{eq:basis}) that
\begin{equation}
u_{ij}(s)\cdot e_{I}=\left\{\begin{array}{ll}
e_{I} & j\notin I  \\
e_{I}+(-1)^{|I\cap(i,j)|}s e_{I\cup\{i\}\setminus\{j\}} & j\in I,i\notin I  \\
e_{I} & i,j\in I  \\
\end{array}\right.
\end{equation}
The first and last cases are not interesting, but the second one allows us to
build our generating set step by step starting from the above mentioned elements.
Keeping track of the appearing sign could cause some difficulty, but we can
overcome this by letting $e_{abc\ldots}=-e_{bac\ldots}$ etc. and simply
substituting $j$ with $i$ without reordering the indices.

Observe, that when $u_{ij}(s)$ acts on a degree $m$ polynomial in the $e_{I}$-s,
then we get a polynomial in $s$ with coefficients in $S^{m}(\bigwedge^{k}\mathcal{H})$.
Since $W$ contains this polynomial for any $s\in\mathbb{C}$, and it is a linear
subspace, $W$ must also contain the coefficient of $s^{l}$ for each $0\le l\le m$
(because the non-vanishing of a Vandermonde determinant). Using this method,
one can calculate in a few steps a generating set for the isotypic (in fact,
irreducible) subspace corresponding to the highest weight. The following lemma
shows which terms should one concentrate on:
\begin{lem}
Let $W\le S^{m}(\bigwedge^{k}\mathcal{H})$ be an invariant subspace and
$S\subseteq W$ a \emph{good} subset

Suppose that $w\in S$ and $i\neq j$ are indices such that $i$ does not appear in
$w$ when written in the monomial basis as above. Then
\begin{enumerate}[a)]
\item The coefficients of every power of $s$ in $u_{ij}(s)\cdot w$ as a polynomial
in $s$ are weight vectors.
\item If the degree of this polynomial is $d$ then the one dimensional subspaces
spanned by the coefficients of $s^{r}$ and $s^{d-r}$ are in the same $S_{n}$-orbit.
\item The coefficient of the constant and the leading terms is contained in $\mathbb{C}S$.
\item If $\mathbb{C}w=\mathbb{C}\pi\cdot w'$ for some $\pi\in S_{n}$, then
the minimal \emph{good} subsets containing $S$ and each coefficient in
the polynomial $u_{ij}(s)\cdot w$ or $u_{\pi^{-1}(i)\pi^{-1}(j)}(s)\cdot w'$ generate the
same subspace.
\end{enumerate}
\end{lem}
\begin{proof}
\begin{enumerate}[a)]
\item $u_{ij}(s)\cdot e_{I_1}e_{I_2}\ldots e_{I_m}=(e_{I_1}+s e_{I_1'})(e_{I_2}+s e_{I_2'})\ldots(e_{I_m}+s e_{I_m'})$
where $I_l'$ is obtained from $I_l$ by replacing $j$ with $i$ if $I_l$ contains
$j$ and $e_{I_l'}=0$ else. The coefficient of $s^l$ contains exactly those terms
in the expansion in which the number of replaced $j$-indices is $l$.
\item $d$ is the (common) number of occurrencies of the index $j$ in each term
of $w$. The coefficient of $s^{d-r}$ term is therefore proportional to the image
of the coefficient of $s^{r}$ under the transposition swapping $e_i$ and $e_j$.
\item The constant term is $w$. 
\item Let $\pi\in S_{n}\le U(\mathcal{H})$ be an element such that $\mathbb{C}w=\mathbb{C}\pi\cdot w'$.
Then
\begin{equation}
\begin{split}
\mathbb{C}u_{ij}(s)\cdot w
  & = u_{ij}(s)\mathbb{C}w  \\
  & = u_{ij}(s)\mathbb{C}\pi\cdot w'  \\
  & = \mathbb{C}u_{ij}(s)\pi\cdot w'  \\
  & = \mathbb{C}\pi u_{\pi^{-1}(i)\pi^{-1}(j)}(s)\cdot w'
\end{split}
\end{equation}
\end{enumerate}
\end{proof}

\begin{cor}
If $S\subseteq S^{m}(\bigwedge^{k}\mathcal{H})$ is a \emph{good} subset and $w\in S$ such that 
in each term of $w$ the index $j$ appears exactly once and $w$ does not contain the index $i$,
then $u_{ij}(s)\cdot w\in\langle S\rangle$ for all $s\in\mathbb{C}$.
\end{cor}
\begin{proof}
In this case, $u_{ij}(s)\cdot w$ is a degree $1$ polynomial in $s$, therefore, by the lemma above,
both of its terms are in $\mathbb{C}S$, hence their sum is in $\langle S\rangle$.
\end{proof}

To sum up, we begin with the vector $e_{12\ldots k}^m$, then act on it and the
distinct types of obtained coefficients of $s$ successively with the matrices
$u_{ij}(s)$, as long as we get new types of vectors. Finally, we take union of
the $S_n$-orbits of the vectors we have met. This will result in a generating set of $W$.

Once we have a generating set, we orthogonalize it, and for each vector $w$ in
the orthogonal set we calculate the value of $|\langle w,\psi^m\rangle|^2\|w\|^{-2}$.
Finally, the sum of these numbers is the value of our invariant evaluated on
the state $\psi$. Explicitely, suppose, that $\psi=\sum_{I}\psi_{I}e_{I}$, and
$w=\sum_{k_1,k_2,\ldots,k_d}\beta_{k_1,k_2,\ldots,k_d}e_{I_1}^{k_1}\ldots e_{I_d}^{k_d}$,
where $d=\binom{n}{k}$ and $I_1,\ldots,I_d$ are the possible $k$-element subsets
of $[n]$, and $k_1,\ldots,k_d$ run over nonnegative integers such that their sum
equals $m$. Then by equation~(\ref{eq:norm})
\begin{equation}\label{eq:innerp}
\begin{split}
\langle w,\psi^m\rangle
  & = \left\langle \sum_{k_1,\ldots,k_d}\beta_{k_1,\ldots,k_d}e_{I_1}^{k_1}\ldots e_{I_d}^{k_d},\sum_{k_1',\ldots,k_d'}\binom{m}{k_1',\ldots,k_d'}\psi_{I_1}^{k_1'}\ldots\psi_{I_d}^{k_d'}e_{I_1}^{k_1'}\ldots e_{I_d}^{k_d'}\right\rangle  \\
  & = \sum_{k_1,\ldots,k_d}\overline{\beta_{k_1,\ldots,k_d}}\psi_{I_1}^{k_1'}\ldots\psi_{I_d}^{k_d'}\underbrace{\binom{m}{k_1,\ldots,k_d}\|e_{I_1}^{k_1}\ldots e_{I_d}^{k_d}\|^2}_{1}
\end{split}
\end{equation}

We would like to remark that if we are to use these invariants as measures of
entanglement, then, taking into account the constraint (\ref{eq:constraint}) and the
fact that the $m$th power of a decomposable state is always in the irreducible
subspace generated by $e_{12\ldots k}^m$, we should use $1-\langle\psi^m,P_{W}\psi^m\rangle=\langle\psi^m,P_{W^{\perp}}\psi\rangle$,
or the invariants associated to the subspaces other than $W$.

If we wanted to calculate the projectors of the other isotypic subspaces, then
we simply needed to take the orthogonal complement of $W$, and find the weight
vectors corresponding to the highest weight, and proceed with it the same way
as we did with $e_{12\ldots k}^m$.

\section{Examples}\label{sec:ex}

\subsection{$k=m=2$ case}

The first nontrivial case is the space of quadratic polynomials in vectors of
the space of two fermions. As we have seen, a weight vector with maximal weight
is $e_{12}^2$, therefore $W:=\langle GL(n,\mathbb{C})e_{12}^2\rangle$ contains $e_{ij}^2$
for $1\le i<j\le n$. In the next step we let $u_{kj}(s)$ act on an element:
\begin{equation}
u_{kj}(s)(e_{ij}^2)=(e_{ij}+s e_{ik})^{2}=e_{ij}^2+2s e_{ij}e_{ik}+s^{2}e_{ik}^2
\end{equation}
This shows that we must add $e_{ij}e_{ik}$ for each triple $i,j,k$, where the
appearing indices are distinct. Now as
\begin{equation}
u_{li}(s)(e_{ij}e_{ik})=(e_{ij}+s e_{lj})(e_{ik}+s e_{lk})=e_{ij}e_{ik}+s(e_{lj}e_{ik}+e_{ij}e_{lk})+e_{lj}e_{lk}
\end{equation}
we also have to add $e_{ij}e_{lk}+e_{ik}e_{lj}$ for each combination of indices.

By the corollary after the lemma, we are ready, but it is instructive to verify
the dimension of the generated subspace. Clearly, $\{e_{ij}^{2}\}_{1\le i<j\le n}\cup\{e_{ij}e_{ik}\}_{1\le i\le n,i\neq j<k\neq i}$ consists of pairwise orthogonal elements. The third type in the
generating set is $\{e_{ij}e_{lk}+e_{ik}e_{lj}\}$ which is seen to generate a
two dimensional space for each set of four indices, and these subspaces are
pairwise orthogonal and also orthogonal to the other elements. Therefore,
\begin{equation}
\dim W=\binom{n}{2}+n\binom{n-1}{2}+2\binom{n}{4}=\frac{n^2(n^2-1)}{12}
\end{equation}
which is exactly the dimension of the irreducible representation of $GL(n,\mathbb{C})$
corresponding to the partition $(2,2)$.

Orthogonalization needs to be performed only within the two dimensional subspaces,
and this leads to the vectors $e_{ij}e_{kl}+e_{ik}e_{jl}$ and $e_{ij}e_{lk}+2e_{il}e_{jk}-e_{ik}e_{jl}$
for $1\le i<j<k<l\le n$. The expression for the invariant corresponding to $W$ is
therefore (using equation~(\ref{eq:innerp}))
\begin{equation}
\begin{split}
I_{(2,2)}(\psi)  & = \langle \psi,P_{W}\psi\rangle \\
  & = \sum_{1\le i<j\le n}|\psi_{ij}^2|^2 + \sum_{i=1}^{n}\sum_{\substack{1\le j<k\le n \\ j\neq i\neq k}}2|\psi_{ij}\psi_{ik}|^2  \\
  & + \sum_{1\le i<j<k<l\le n}\left(|\psi_{ij}\psi_{kl}+\psi_{ik}\psi_{jl}|^2+\frac{1}{3}|\psi_{ij}\psi_{lk}+2\psi_{il}\psi_{jk}-\psi_{ik}\psi_{jl}|^2\right)
\end{split}
\end{equation}

In this case, we can also show that $W^{\perp}$ is irreducible. To this end,
let us recall that for $n=4$ there exists a degree two $SL(4,\mathbb{C})$-invariant
over $\bigwedge^{2}\mathbb{C}^{4}$, namely, the polynomial in the Pl\"ucker relation
which is known to be a sufficient and necessary condition of separability. The
subrepresentation generated by this polynomial is the representation indexed
by the partition $(1,1,1,1)$, therefore this one must appear also in the $n\neq 4$
case. As the dimension of this is $\binom{n}{4}$, and
\begin{equation}
\dim W+\binom{n}{4}=\frac{n(n-1)(n^2-n+2)}{8}=\dim S^{2}(\bigwedge^{2}\mathbb{C}^{n})
\end{equation}
therefore $W^{\perp}$ is irreducible, and the unitary invariant associated to it
gives a generalization of the Pl\"ucker relation. The explicit formula turns out
to be simpler than the previous one:
\begin{equation}
I_{(1,1,1,1)}(\psi) = \langle \psi,P_{W^{\perp}}\psi\rangle = \sum_{1\le i<j<k<l\le n}\frac{2}{3}|\psi_{ij}\psi_{kl}+\psi_{ik}\psi_{lj}+\psi_{il}\psi_{jk}|^2
\end{equation}

\subsection{$k=2$, $m=3$ case}

In this case, a weight vector for the highest weight is $e_{12}^{3}$. Again,
$W:=\langle GL(n,\mathbb{C})e_{12}^{3}\rangle$. We are looking for a generating
set of $W$. We extend $e_{12}^3$ into a \emph{good} set $\{e_{ij}\}_{1\le i<j\le n}$.
Now we need to add the coefficient of $s$ in
\begin{equation}
u_{32}(s)(e_{12}^3)=(e_{12}+s e_{13})^3=e_{12}+3s e_{12}^2e_{13}+3s^2 e_{12}e_{13}^2+s^3e_{13}^3
\end{equation}
and one vector from each element of the orbit of the subspace generated by it:
$\{e_{ij}^2e_{ik}\}_{i,j,k\in[n]}$. The next steps are:
\begin{equation}
u_{43}(s)(e_{12}e_{13}^2)=e_{12}(e_{13}+se_{14})^2=\ldots+2se_{12}e_{13}e_{14}+s^2(\ldots)
\end{equation}
\begin{equation}
\begin{split}
u_{m1}(s)(e_{12}e_{13}^2) & =(e_{12}+se_{m2})(e_{13}+se_{m3})^2  \\
  &  =\ldots+s(2e_{12}e_{13}e_{m3}+e_{m2}e_{13}^{2})+s^2(\ldots)+s^3(\ldots)
\end{split}
\end{equation}
Here $m=2$ is special, in this case the second term in the coefficient of $s$
vanishes, hence we have to add $\{2e_{ij}e_{ik}e_{mk}+e_{mj}e_{ik}^{2}\}$ for any
ordered pair of disjoint pairs $({i,j},{k,m})$, and also $e_{ij}e_{jk}e_{kj}$ for
$\{i,j,k\}\in\binom{[n]}{3}$. The remaining steps are
\begin{equation}
\begin{split}
u_{m1}(s)e_{12}e_{13}e_{14} & =(e_{12}+se_{m2})(e_{13}+se_{m3})(e_{14}+se_{m4})  \\
  & =\ldots+s(e_{m2}e_{13}e_{14}+e_{12}e_{m3}e_{14}+e_{12}e_{13}e_{m4})+s^2(\ldots)+s^3(\ldots)
\end{split}
\end{equation}
\begin{equation}\label{eq:sixterm}
\begin{split}
u_{m1}(s)(e_{52}e_{13}e_{14}+e_{12}e_{53}e_{14}+e_{12}e_{13}e_{54}) & =\ldots+s(e_{52}e_{m3}e_{14}+e_{52}e_{13}e_{m4}  \\
  &  +  e_{m2}e_{53}e_{14}+e_{12}e_{53}e_{m4}  \\
  &  +  e_{m2}e_{13}e_{54}+e_{12}e_{m3}e_{54})+s^2(\ldots)
\end{split}
\end{equation}
Here $m\le 5$ does not lead to a new subspace.

It turns out that the vectors obtained so far are enough to generate $W$. In this case,
orthogonalization turns out to be a bit lengthy, especially in the case of the six-term
vectors like in equation (\ref{eq:sixterm}). These span a five dimensional subspace for
each six-element set of indices $i_1,\ldots,i_6$. For these the coefficients of the
monomials are given as a matrix:
\begin{equation}
\left(\begin{array}{ccccccccccccccc}
 0 &  0 &  0 &  0 &  0 &  0 &  0 &  1 &  1 &  0 &  1 &  1 &  0 &  1 &  1 \\
 0 &  0 &  0 &  0 &  3 &  3 &  0 &  1 &  1 &  3 &  1 & -2 &  3 &  1 & -2 \\
 0 &  0 &  0 &  2 &  1 & -1 &  2 &  1 & -1 &  1 &  1 &  0 & -1 & -1 &  0 \\
 0 &  2 &  2 &  0 &  1 &  1 &  0 &  1 &  1 & -1 & -1 &  0 & -1 & -1 &  0 \\
 4 &  2 & -2 &  2 &  1 & -1 & -2 & -1 &  1 & -1 &  1 &  2 &  1 & -1 & -2 \\
\end{array}\right)
\end{equation}
The order of the monomials is
$(12)(34)(56)$, $(12)(35)(46)$, $(12)(36)(45)$, $(13)(24)(56)$, $(13)(25)(46)$, $(13)(26)(45)$, $(14)(23)(56)$, $(14)(25)(36)$, $(14)(26)(35)$, $(15)(23)(46)$, $(15)(24)(36)$, $(15)(26)(34)$, \\ $(16)(23)(45)$, $(16)(24)(35)$, $(16)(25)(34)$, where $(ab)(cd)(ef)$ is a short notation for $e_{i_a,i_b}e_{i_c,i_d}e_{i_e,i_f}$.
The norms inverse squared of these vectors are
\begin{equation}
1, \frac{1}{8}, \frac{3}{8}, \frac{3}{8}, \frac{1}{8}
\end{equation}
respectively. The orthogonal generators coming from the remaining vectors are
given in table (\ref{tab:ortho})

\begin{table}[htb]
\centering
\begin{tabular}{p{5.5cm}|c|c|c}
form  &  indices  &  dimension  &  $\|\cdot\|^{-2}$  \\
\hline
$e_{ij}^3$  &  $\{i,j\}$  &  $\binom{n}{2}$  &  $1$  \\
\hline
$e_{ij}^2e_{ik}$  &  $\{i\},\{j\},\{k\}$  &  $n(n-1)(n-2)=6\binom{n}{3}$  &  $3$  \\
\hline
$e_{ij}e_{ik}e_{il}$  &  $\{i\},\{j,k,l\}$  &  $n\binom{n-1}{3}=4\binom{n}{4}$  &  $6$  \\
\hline
$e_{ij}e_{ik}e_{kj}$  &  $\{i,j,k\}$  &  $\binom{n}{3}$  &  $6$  \\
\hline
$e_{ij}^2e_{kl}+2e_{ij}e_{il}e_{kj}$  &  $\{i,j\},\{k,l\}$  &  $2\binom{n}{2}\binom{n-2}{2}$  &  $1$  \\
\cline{1-1} \cline{4-4}
$-2e_{ij}^2e_{kl}+6e_{ij}e_{ik}e_{lj}+2e_{ij}e_{il}e_{kj}$  &    &    &  $\frac{1}{8}$  \\
\hline
$e_{ik}e_{il}e_{jm}+e_{ik}e_{jl}e_{im}+e_{jk}e_{il}e_{im}$  &  $\{i\},\{j,k,l,m\}$  &  $3n\binom{n-1}{4}=15\binom{n}{5}$  &  $2$  \\
\cline{1-1} \cline{4-4}
$e_{ik}e_{il}e_{jm}+3e_{ik}e_{ij}e_{ml}+e_{im}e_{il}e_{kj}+3e_{im}e_{ij}e_{kl}+2e_{im}e_{ik}e_{jl}$  &    &    &  $\frac{1}{4}$  \\
\cline{1-1} \cline{4-4}
$2e_{ij}e_{il}e_{mk}+e_{ik}e_{il}e_{mj}+e_{ik}e_{ij}e_{ml}+e_{im}e_{il}e_{jk}+e_{im}e_{ij}e_{kl}$  &    &    &  $\frac{3}{4}$  \\
\end{tabular}
\caption{Orthogonalized generators for $W$. Indices shown in one set are indistinguishable for counting purposes.\label{tab:ortho}}
\end{table}

Using these data, the value of the invariant $I_{(3,3)}$ can be calculated in a straightforward
way, but the full formula is too long to be presented explicitely.

The orthogonal complement of $W$ clearly has a highest weight of $(2,2,1,1)$,
and we could find a generator of the unique one dimensional weight space corresponding
to it, and calculate the projector of its invariant subspace. Instead of this,
we follow another approach. According to the plethysm
\begin{equation}
s_{(3)}[s_{(1,1)}]=s_{(3,3)}+s_{(2,2,1,1)}+s_{(1,1,1,1,1,1)}
\end{equation}
for $n=6$, an $SL(6,\mathbb{C})$-invariant polynomial appears. It is easy to
guess how this should look like: for a state $\psi\in\bigwedge^{2}\mathbb{C}^6$,
we can construct $\psi\wedge\psi\wedge\psi$ which is an element of $\bigwedge^{6}\mathbb{C}^6$,
a one dimensional vector space on which $GL(6,\mathbb{C})$ acts by multiplication
with the determinant. Therefore this element remains unchanged under $SL(6,\mathbb{C})$,
and its norm squared is an $U(6,\mathbb{C})$-invariant polynomial in the coefficients
of $\psi$ and their conjugates. Our invariant corresponding to the subrepresentation
indexed by the partition $(1,1,1,1,1,1)$ must be proportional to it. Explicitely,
it equals to
\begin{equation}
\frac{1}{11520}\left|\sum_{\pi\in S_6}\sigma(\pi)\psi_{\pi(1),\pi(2)}\psi_{\pi(3),\pi(4)}\psi_{\pi(5),\pi(6)}\right|^{2}
\end{equation}
Here, the sum is over all the permutations, but actually there are $15$ different
terms, each counted $48=3!\cdot 2^3$ times. Alternatively, we could sum over
the partitions of $[n]$ into three two-element sets.

The $n\ge 6$ case can be obtained similarly to the previous section. Taking all
the six-element subsets of $[n]$ polynomials like this span an $\binom{n}{6}$
dimensional subspace which is also the dimension of the invariant subspace we
are looking for. Therefore in the general case the invariant is
\begin{equation}
I_{(1,1,1,1,1,1)}(\psi)=\frac{1}{11520}\sum_{I\in\binom{[n]}{6}}\left|\sum_{\pi\in S_6}\sigma(\pi)\psi_{i_\pi(1),i_\pi(2)}\psi_{i_\pi(3),i_\pi(4)}\psi_{i_\pi(5),i_\pi(6)}\right|^{2}
\end{equation}
where $I=\{i_1,\ldots,i_6\}$.

These two invariants are linearly independent, and they sum to $1$ with the one
associated to the third irreducible subspace.

\subsection{$k=3$, $m=2$ case}

Now we turn to the first case with more than two particles. In
$S^{2}(\bigwedge^{3}\mathcal{H})$ the vector with highest weight is $e_{123}^2$.
We proceed in a similar way as before:
\begin{equation}
u_{n3}(s)(e_{123}^2)=e_{123}^2+2se_{123}e_{12n}+s^{2}e_{12n}^2
\end{equation}
\begin{equation}
u_{n2}(s)(e_{123}e_{124})=\ldots+s(e_{123}e_{1n4}+e_{1n3}e_{124})+s^2(\ldots)
\end{equation}
\begin{equation}\label{eq:fourterm}
u_{n1}(s)(e_{123}e_{154}+e_{153}e_{124})=\ldots+s(e_{123}e_{n54}+e_{n23}e_{154}+e_{153}e_{n24}+e_{n53}e_{124})+s^2(\ldots)
\end{equation}
These vectors already form a generating set, we only need to orthogonalize this
set. For a fixed subset of six indices, the vectors of the form like in (\ref{eq:fourterm})
span a five dimensional subspace. Orthogonal generators for this are again given with
the coefficients of the monomials as a matrix:
\begin{equation}\label{eq:coeff2}
\left(\begin{array}{cccccccccc}
 1 &  0 &  0 &  1 &  0 &  1 &  0 &  0 &  0 &  1  \\
-1 &  0 &  2 &  1 &  0 &  1 &  2 &  0 &  0 & -1  \\
-1 &  0 &  0 & -1 &  2 &  1 &  0 &  0 &  2 &  1  \\
 1 &  4 &  2 & -1 &  2 &  1 & -2 &  0 & -2 & -1  \\
 1 &  1 & -1 & -1 & -1 &  1 &  1 &  3 &  1 & -1  \\
\end{array}\right)
\end{equation}
The order of monomials is
$(123)(456)$, $(124)(356)$, $(125)(346)$, $(126)(345)$, $(134)(256)$, $(135)(246)$, $(136)(245)$, $(145)(236)$, $(146)(235)$, $(156)(234)$, where $(abc)(def)$ is a shorthand notation for the vector $e_{i_a,i_b,i_c}e_{i_d,i_e,i_f}$.
The inverse squared norms of these vectors are
\begin{equation}
\frac{1}{2}, \frac{1}{6}, \frac{1}{6}, \frac{1}{18}, \frac{1}{9}
\end{equation}
respectively. The orthogonal generators coming from the remaining vectors are
given in table~(\ref{tab:ortho2})

\begin{table}[htb]
\centering
\begin{tabular}{p{5.5cm}|c|c|c}
form  &  indices  &  dimension  &  $\|\cdot\|^{-2}$  \\
\hline
$e_{ijk}^2$  &  $\{i,j,k\}$  &  $\binom{n}{3}$  &  $1$  \\
\hline
$e_{ijk}e_{ijl}$  &  $\{i,j\},\{k,l\}$  &  $\binom{n}{2}\binom{n-2}{2}$  &  $2$  \\
\hline
$e_{ijm}e_{ikl}+e_{ijk}e_{iml}$  &  $\{i\},\{j,k,l,m\}$  &  $2n\binom{n-1}{4}$  &  $1$  \\
\cline{1-1} \cline{4-4}
$e_{ijm}e_{ikl}+2e_{ijl}e_{imk}-e_{ijk}e_{iml}$  &    &    &  $\frac{1}{3}$  \\
\end{tabular}
\caption{Orthogonalized generators for the subspace generated by the highest weight vector. Indices shown in one set are indistinguishable for counting purposes.\label{tab:ortho2}}
\end{table}

The value of $I_{(2,2,2)}$ can now be calculated. This time the orthocomplement
is also irreducible, so we get one independent invariant in this case.

\section{SLOCC-invariants and local unitary invariants}\label{sec:SLOCC}

In the examples we have seen local unitary invariants with a special property: for
a particular value of $n$, the corresponding irreducible subspace becomes one
dimensional, and the subspace is pointwise fixed under the action of $SL(n,\mathbb{C})$,
that is, the local unitary invariant turns out to be a SLOCC-invariant. Let us examine
this case in more detail.

The irreducible polynomial representation of $SL(n,\mathbb{C})$ indexed by the
partition $\lambda$ is one dimensional precisely when $\lambda$ consists of equal
parts. In this case, $\lambda=(r,r,\ldots,r)$ is a partition of $nr$, hence a
neccesary condition for it to occur as a subrepresentation of $S^{m}(\bigwedge^{k}\mathbb{C}^{n})$
is that $mk=nr$, and in this case $GL(n,\mathbb{C})$ acts on it by multiplication
with the $r$th power of the determinant. The norm squared is therefore invariant
under $U(n,\mathbb{C})$.

In our notations, this subspace is spanned by a polynomial $w$ in the basis
vectors $e_{1},\ldots,e_{n}$. $w$ is a weight vector with weight $(r,r,\ldots,r)$,
and it generates a one dimensional $U(n,\mathbb{C})$-invariant subspace. The
crucial thing is that when we increase the dimension $n$ of the single particle
state space to $n'$, $w$ remains a weight vector that generates an irreducible
$U(n',\mathbb{C})$-invariant subspace, but it is no longer one dimensional.
Therefore, the invariant corresponding to this subspace will be a generalization
of the SLOCC-invarant we have begun with, but is now only a unitary invariant.

The explicit form of the resulting invariant can be obtained in general using
the method outlined above: we must act on it with $u_{ij}(s)$-s and elements
of $S_{n'}$. A particularly simple special case is when $r=1$. In this case
the dimension of the representation corresponding to $\lambda$ is $\binom{n'}{n}$,
and an orthonormal basis can be obtained by acting on $w$ by elements of $S_{n'}$.
Therefore the invariant can be obtained by calculating the value of the
SLOCC-invariant with the initial index set $[n]$ replaced by every element
of $\binom{[n']}{n}$, and summing their absolute values squared.

\section{Distinguishable particles}\label{sec:dist}

We can also obtain many (but not all) local unitary invariants for distinguishable
subsystems from our fermionic ones\cite{Gittings,Eckert,VL}. Suppose that we would like to find local
unitary invariants for a quantum system containing $k$ subsystems with Hilbert
space dimensions $n_1,\ldots,n_k$, and let $n=n_1+\ldots+n_k$. Then using the
branching rule
\begin{equation}
\begin{split}
U(n,\mathbb{C}) & \ge U(n_1,\mathbb{C})\times\ldots\times U(n_k,\mathbb{C})  \\
\bigwedge^{k}(\mathbf{n}) & \to \bigwedge^{k}((\mathbf{n_1})\oplus\ldots\oplus(\mathbf{n_k}))\simeq\ldots\oplus(\mathbf{n_1},\mathbf{n_2},\ldots,\mathbf{n_k})\oplus\ldots
\end{split}
\end{equation}
we can identify the full state space with a subspace of a fermionic Hilbert space
with $k$ particles and $n$ dimensional single particle state space. Then we can
pull back the fermionic invariants constructed by the method above.
For example, for $k=2$, the restriction of $I_{(1,1,1,1)}$ is proportional
to the square of the norm of the concurrence vector\cite{Horodecki}.

Observe, that these special LU-invariants have a larger symmetry group, in
particular, if $n_1=\ldots=n_{k}$, then they are also permutation invariant.

As an other example, let us consider the quantum system of three qubits. Using
the scheme outlined above, its Hilbert space can be thought of as a subspace
of the Hilbert space of a three-fermion system with six single particle states.
In the previous section (see equation~(\ref{eq:coeff2}) and table~\ref{tab:ortho2})
we have derived a formula for a local unitary invariant
which can be restricted to this subspace. The restriction turns out to be the
following:
\begin{equation}
\begin{split}
I_{(2,2,2)}(\psi) & = \sum_{i,j,k=0}^{1}|\psi_{ijk}^2|^2+2\sum_{i,j=0}^{1}\left(|\psi_{ij0}\psi_{ij1}|^2+|\psi_{i0j}\psi_{i1j}|^2+|\psi_{0ij}\psi_{1ij}|^2\right)  \\
& + \sum_{i=0}^{1}\left(|\psi_{i01}\psi_{i10}|^2+\frac{1}{3}|\psi_{i01}\psi_{i10}+2\psi_{i00}\psi_{i11}|^2+|\psi_{0i1}\psi_{1i0}|^2\right.  \\
& \left.+\frac{1}{3}|\psi_{0i1}\psi_{1i0}+2\psi_{0i0}\psi_{1i1}|^2+|\psi_{01i}\psi_{10i}|^2+\frac{1}{3}|\psi_{01i}\psi_{10i}+2\psi_{00i}\psi_{11i}|^2\right)  \\
& + \frac{1}{2}|\psi_{000}\psi_{111}|^2+\frac{1}{6}|\psi_{000}\psi_{111}+2\psi_{001}\psi_{110}|^2+\frac{1}{6}|\psi_{000}\psi_{111}+2\psi_{011}\psi_{100}|^2  \\
& + \frac{1}{18}|\psi_{000}\psi_{111}-2\psi_{001}\psi_{110}-2\psi_{011}\psi_{100}|^2  \\
& + \frac{1}{9}|\psi_{000}\psi_{111}+\psi_{001}\psi_{110}+3\psi_{010}\psi_{101}+\psi_{011}\psi_{100}|^2
\end{split}
\end{equation}
where $\psi=\sum_{i,j,k=0}^{1}\psi_{ijk}e_i\otimes e_j\otimes e_k$ is a three-qubit
state. This quantity is invariant under local unitary transformations and permutations
of the three subsystems. Plugging in the coefficients of certain states into this
formula reveals that for a separable state the value is $1$ as expected, for the GHZ-state
it is $\frac{3}{4}$, for the W-state $\frac{7}{9}$, and for a biseparable state
with maximal entanglement shared by two qubits, the value is $\frac{5}{6}$. This
behaviour is in contrast with the well-known three-qubit invariant, Cayley's
hyperdeterminant, which is only sensitive to GHZ-type entanglement.

Note that these values are one minus the quarter of the values of the entanglement
measure proposed in \cite{Meyer} for three qubits, so one is tempted to conjecture
that the restriction of $4-4I_{(2,2,2)}$ equals to their measure. In their paper
this measure appears as a member of an infinite family of multiqubit entanglement
measures. It would be interesting to see whether every member can be found using
our method. If this is the case, then we could generalize them to arbitrary
dimensional single particle states and also to fermionic systems.

A more detailed treatment of this method of constructing entanglement measures
for distinguishable subsystems from fermionic measures along with more examples
can be found in \cite{VL}.

\section{Conclusion}

In this paper a certain class of entanglement measures for fermionic quantum
systems has been introduced and studied. A way to obtain their explicit form
is presented, and it was pointed out that this form is independent of the
dimension of the single particle state space. Some examples are discussed in
detail. The connection to SLOCC-invariants and the case of distinguishable
subsystems is also mentioned.

At this point a number of questions arise. Further study is needed to explore
the behaviour of these local unitary invariants. For instance, are there any
entanglement monotones among them, and is there a way to characterize these?
Could we find the convex roof extension of any of them? Is there a physical
meaning for these quantities, in what way do they measure entanglement?

Also, similar method could be applied directly to quantum systems with
distinguishable particles, or to mixed states. It would be interesting to
see if one could obtain entanglement measures this way which can be useful
in practice.

\section{Appendix: Representation theory of the unitary group}

In this appendix, the basic aspects of the representation theory of the unitary
groups is summarized. These and many more can be found in many textbooks, see
e.g. \cite{FH}.

The unitary group $U(n,\mathbb{C})$ consists of $n\times n$ complex matrices
satisfying $A^{*}=A^{-1}$. Let $T\le U(n,\mathbb{C})$ be the subgroup of
diagonal matrices, $T$ is a maximal torus in $U(n,\mathbb{C})$. We are interested
solely in finite dimensional polynomial representations of $U(n,\mathbb{C})$
which are determined by their characters which are in turn uniquely encoded in
the restriction of the characters to $T$. The characters themselves are symmetric
polynomials in the eigenvalues (the diagonal elements, in the case of $T$).

For example, the action of $U(n,\mathbb{C})$ on the $n$ dimensional vector
space $C^n$ of column vectors by left multiplication is called the standard
representation. Its character is represented by symmetric polynomial
$x_1+x_2+\ldots+x_n$.

A representation $V$ of $U(n,\mathbb{C})$ is also a representation of $T$ whose
irreducible representations are one dimensional, and hence can be considered to
be homomorphisms $T\to\mathbb{C}^{\times}$. An irreducible representation
of $T$ with nonzero multiplicity in $V$ is called a weight. The weight space
corresponding to weight is the union of the subrepresentations in $V$ isomorphic
to a given weight.

Continuing the previous example, the weights of the standard representation are
the representations $\rho_{k}:diag(\lambda_1,\ldots,\lambda_n)\mapsto \lambda_k$,
and the weight space of $\rho_{k}$ is spanned by the $k$th standard basis element.
From now on, the product $\rho_{1}^{r_1}\otimes\ldots\otimes\rho_{n}^{r_n}$ will be
simply denoted by $(r_1,\ldots,r_n)$.

There is the usual partial ordering on the set of weights which we identify with $\mathbb{Z}^n$:
$(r_1,r_2,\ldots,r_n)$ is called positive iff $r_1+\ldots+r_n=0$ and $r_1$, $r_1+r_2$,
\ldots, $r_1+r_2+\ldots+r_{n-1}$ are nonnegative, and $\lambda\ge\mu$ iff $\lambda-\mu$
is positive. The isomorphism class of an irreducible representation of $U(n,\mathbb{C})$
is determined by its highest weight and the weight space for the highest weight
is one dimensional. Using this fact, we can decompose any finite dimensional
representation of $U(n,\mathbb{C})$ into the direct sum of irreducible representations.

The dimension of an irreducible representation of $U(n,\mathbb{C})$ with highest
weight $\lambda=(\lambda_1,\ldots,\lambda_n)$ is
\begin{equation}
\dim S_{(\lambda_1,\ldots,\lambda_{n})}V=\prod_{1\le i<j\le n}\frac{\lambda_i-\lambda_j+j-i}{j-i}
\end{equation}
Note that if we start to increase $n$ and pad $\lambda$ with zeros on the right,
then the dimension in the function of $n$ turns out to be a polynomial of degree
$\lambda_1+\ldots+\lambda_n$. In fact the entries of the representing matrices
are polynomials in the entries of the represented matrix with this same degree.
The symmetric polynomial giving the character of this representation is the
Schur-polynomial.

Calculating symmetric and exterior powers, tensor products, decomposition to
irreducible subrepresentations and many more can be done working only with
symmetric polynomials. This is a crucial fact for our purposes, as a symmetric
polynomial of degree $d$ is determined by its terms containing unknowns only
from a fixed set of $d$ variables. Indeed, no term can contain more than $d$
variables, so the missing ones are obtained by permutations of the variables.
This fact implies that the decomposition of $S^m(\bigwedge^{k}\mathbb{C}^n)$ to
irreducible $U(n,\mathbb{C})$-representations is independent of $n$ when $n\le km$.

\end{document}